\documentclass[a4paper, 11pt]{amsart}
\usepackage{amssymb,array}
\usepackage{url}
\usepackage{graphicx}
\usepackage{subfig}
\usepackage[left=1in, right=1in, top=1in, bottom=1in]{geometry}

\renewcommand{\leq}{\leqslant}
\renewcommand{\geq}{\geqslant}

\newcommand{\N}{\mathbb{N}}

\newcommand{\C}{\mathbb{C}}
\newcommand{\E}{\mathbb{E}}

\newcommand{\U}{\mathcal{U}}

\renewcommand{\S}{\mathcal{S}}

\newcommand{\bra}[1]{\langle #1 |}
\newcommand{\ket}[1]{| #1 \rangle}

\renewcommand{\d}[1]{\mathrm{d}#1}
 
\newcommand{\ol}{\overline}

\newcommand{\isom}{\simeq}

\DeclareMathOperator{\trace}{Tr}

\DeclareMathOperator{\id}{id}
\DeclareMathOperator{\Wg}{Wg}
\DeclareMathOperator{\Mob}{Mob}

\DeclareMathOperator*{\Ex}{\mathbb{E}}

\renewcommand{\phi}{\varphi}
\newcommand{\iy}{\infty}

\newtheorem{theorem}{Theorem}[section]
\newtheorem{definition}[theorem]{Definition}
\newtheorem{proposition}[theorem]{Proposition}
\newtheorem{remark}[theorem]{Remark}
\newtheorem{lemma}[theorem]{Lemma}

\def\l{\left}
\def\r{\right}
\def\bra{\langle}
\def\ket{\rangle}

\def\be{\begin{eqnarray}}
\def\ee{\end{eqnarray}}
\def\bee{\begin{equation}}
\def\eee{\end{equation}} 

\begin{document}

\title[Towards states achieving MOE for product channels]{Towards a state minimizing the output entropy of a tensor product of random quantum channels}
\author{Beno\^{\i}t Collins}
\address{
D\'epartement de Math\'ematique et Statistique, Universit\'e d'Ottawa,
585 King Edward, Ottawa, ON, K1N6N5 Canada, 
CNRS, Institut Camille Jordan Universit\'e  Lyon 1, 43 Bd du 11 Novembre 1918, 69622 Villeurbanne
France
and 
Research Institute for Mathematical Sciences, Kyoto University, Kyoto 606-8502 JAPAN} 
\email{bcollins@uottawa.ca}
\author{Motohisa Fukuda}
\address{
Institute for Quantum Computing,
University of Waterloo,
200 University Ave. West
Waterloo, Ontario, N2L 3G1 Canada}
\email{mfukuda@uwaterloo.ca}
\author{Ion Nechita}
\address{
CNRS, Laboratoire de Physique Th{\'e}orique , IRSAMC, Universit{\'e} de Toulouse, UPS, F-31062 Toulouse, France}
\email{nechita@irsamc.ups-tlse.fr}
\subjclass[2000]{Primary 15A52; Secondary 94A17, 94A40} 
\keywords{Random matrices, Weingarten calculus, Quantum information theory, Random quantum channel}

\begin{abstract}
We consider the image of some classes of bipartite quantum states under a tensor product of random quantum channels.
Depending on natural assumptions that we make on the states, the eigenvalues of their outputs have new properties which we describe.
Our motivation is provided by the additivity questions in quantum information theory, and we build on the idea that a Bell state sent through a product of conjugated
random channels has at least one large eigenvalue.
We generalize this setting in two directions.
First, we investigate general entangled pure inputs and show that that Bell states give the least entropy among those inputs in the asymptotic limit.
We then study mixed input states, and obtain new multi-scale random matrix models that allow to quantify the difference of 
the outputs' eigenvalues between a quantum channel and its complementary version in the case of a non-pure input.
\end{abstract}

\maketitle

\section{Introduction}
\subsection{Background} 
One of the most important questions in quantum communication theory 
is whether a quantum channel has additive properties or not \cite{bcn,cn2,Fuk06,fukuda-king,fukuda-king-moser,FukudaWolf07,hastings,hayden-winter,Shor04}.
If a channel $\Phi$ is additive for the Holevo capacity $\chi(\cdot)$
in the sense that 
$\exists N \in \N, \, \forall n \geq N $
\be\label{additivity1}
\chi(\Phi^{\otimes n}) = n \chi (\Phi)
\ee
then the classical capacity of the quantum channel equals the Holevo capacity,
giving a one-shot (non-asymptotic) formula for the classical capacity $C(\cdot)$: 
\be
C(\Phi) = \lim_{n\rightarrow \infty} \frac{1}{n} \chi (\Phi^{\otimes n}) = \chi (\Phi)
\ee
By definition,
\be
\chi(\Phi) = \max_{\{p_i, \rho_i\}} \l[S\l( \sum_ip_i\Phi(\rho_i)\r) - \sum_i p_i S\l(\Phi(\rho_i)\r)\r]
\ee
where, $\{p_i, \rho_i\}$ are ensembles and $S(\cdot)$ is the von Neumann entropy.
The above additive property was conjectured to be true for all quantum channels until 
Hastings showed \cite{hastings} existence of channels such that
\be
S_{\min} (\Phi \otimes \bar \Phi) < S_{\min} (\Phi) + S_{\min} (\bar\Phi)  
\ee 
Here, $S_{\min}(\cdot)$ is the minimal output entropy (MOE).  
Indeed, this result also gives counterexamples to (\ref{additivity1}) 
by \cite{Shor04,FukudaWolf07,Fuk06} and
as a result 
$C(\Phi) \not= \chi(\Phi)$ in general.
Note that 
\be\label{eq:MOE}
S_{\min} (\Phi) = \min_{\rho} S(\Phi(\rho)) 
\ee
where 
the minimum is take over all pure inputs (rank-one projections). 

One of the most important results in the additivity theory of quantum channels is Hastings' proof of the additivity of the minimal output entropy \cite{hastings}. 
The proof contains two disjoint parts: establishing a lower bound for the minimum output entropy for one channel, which was the main contribution of Hastings, 
and an upper bound for the minimal output entropy of the product channel. The latter part, a very delicate question, is dealt with in a very crude manner,
by simply bounding the minimum 
over the set of all bi-partite input states with the value of a single sample. 
It is thus of utmost importance to choose a state with small output entropy and Hastings uses an idea introduced by Hayden and Winter in \cite{hayden-winter}. 
This idea, putting a maximally entangled state (or a Bell state) through a product of conjugate channels has been, to our knowledge, the only example of a bi-partite input state 
with small output entropy. 

The purpose of this paper is to generalize this idea and to introduce new classes of interesting input states.  
Our results are derived using mainly the graphical calculus introduced in \cite{cn1} and the technique of moments in random matrix theory. 
Our conclusion is that the Bell state gives asymptotically, among some large classes of input states,  the output with the least entropy. 

Our results do not imply that the Bell state gives the largest violation of additivity - and proving such a result is certainly very difficult as it would require an optimal bi-partite version of
the lower bounds of \cite{hastings} - but they stand as a solid mathematical evidence towards the fact that the physically intuitive choice of the Bell state is indeed close to being optimal.

The paper is organized as follows. 
In Section \ref{sec:review} we recall some basic facts about random quantum channels, the Hayden-Winter trick and the graphical notation needed to perform 
the integrals over the unitary group in the rest of the paper. Sections \ref{sec:generalized-Bell} and \ref{sec:generalized-Bell-UU} generalize the idea of Hayden and Winter, 
both in the case of conjugate channels and in the case of identical channels. In Section \ref{sec:conjugated} we explain why the Hayden-Winter trick only works in the case of 
conjugate channels. Section \ref{sec:mixed-inputs} introduces two models of mixed bi-partite inputs that quantify the difference between a quantum channel and its conjugate 
when the input is not pure.

\section{Review on random quantum channels and unitary integration}\label{sec:review}

\subsection{Random quantum channels}
A quantum channel $\Phi: \mathbb M_{d_{in}}(\mathbb C) \to \mathbb M_{d_{out}}(\mathbb C)$
in the Stinespring's picture is described as
\be
\Phi (\rho) = \trace_{\C^{d_{env}}} \l[V \rho V^* \r]
\ee 
where
\be
V : \C^{d_{in}} \rightarrow \C^{d_{env}} \otimes \C^{d_{out}} 
\ee  
is an isometry. 
Here, $d_{in}$, $d_{out}$ and $d_{env}$ are dimensions of input space, output space and environment respectively. 
Moreover,
we induce the measure on the set of quantum channels from the Haar measure on the unitary group $\U (d_{out}\cdot d_{env})$ in the following way. We endow the set of isometries by truncating a unitary matrix distributed along the Haar measure and we consider the image measure on the set of quantum channels.

If we switch the roles of output and environment spaces 
a quantum channel becomes what is called
 its complementary channel $\Phi^C: \mathbb M_{d_{in}}(\mathbb C) \to \mathbb M_{d_{env}}(\mathbb C)$
\cite{Holevo05,KMNR}:
\be
\Phi^C (\rho) = \trace_{\C^{d_{out}}} \l[V \rho V^* \r]
\ee 
A quantum channel and its complementary channel share 
the same output eigenvalues for any pure input via Schmidt decomposition, see  \cite{Holevo05,KMNR} for details. Our interest in complementary channels is motivated by the fact that, often, the size of the environment is smaller than the output size, so output states are easier to study for $\Phi^C$.

\subsection{The Hayden-Winter trick}
As stated above, 
in proving violation of additivity one needs a small enough upper bound for the minimum output entropy of product of two quantum channels
to show that some entangled input gives an output with strictly less entropy than all the product inputs do.
The idea introduced by Hayden and Winter is that 
the Bell state gives an output with a large eigenvalue via the product of any channel and its complex conjugate. 
More precisely, consider the \emph{maximally entangled} (or the Bell) state
\be
|\phi_m \ket = \frac{1}{d_{in}}\sum_i^{d_{in}} | i \ket |i \ket 
\ee
where $\{|i\ket\}_i$ are the canonical basis vectors. 
Then, it has been shown in \cite{hayden-winter, cn1} that, for any channel $\Phi$ we have
\be \label{HWtrick}
\bra \phi_m | (\Phi \otimes \bar\Phi) (|\phi_m\ket\bra\phi_m|) | \phi_m \ket \geq \frac{d_{in}}{d_{out} \cdot d_{env}}.
\ee
This yields
a lower bound for the largest eigenvalue of the output for $\Phi \otimes \bar\Phi$.
In turn, this can provide a bound that is small enough for the MOE in order to ensure violation of additivity, using the inequality
\begin{equation}
S_{\min}(\Phi \otimes \bar \Phi) \leq S[ (\Phi \otimes \bar\Phi) (|\phi_m\ket\bra\phi_m|)].
\end{equation}

Note that in the above inequality, $\bar\Phi$ (the complex conjugate channel of $\Phi$) is defined by 
replacing $U$ by $\overline U$ in the definition of $\bar \Phi$.  

\subsection{Unitary integration}
  
  Let us start by recalling the definition of a combinatorial object of interest, the unitary Weingarten function.
\begin{definition}
The unitary Weingarten function 
$\Wg(n,\sigma)$
is a function of a dimension parameter $n$ and of a permutation $\sigma$
in the symmetric group $\S_p$. 
It is the inverse of the function $\sigma \mapsto n^{\#  \sigma}$ under the convolution for the symmetric group ($\# \sigma$ denotes the number of cycles of the permutation $\sigma$).
\end{definition}

Notice that the  function $\sigma \mapsto n^{\# \sigma}$ is invertible when $n$ is large, as it
behaves like $n^p\delta_e$
as $n\to\infty$.
Actually, if $n<p$ the function is not invertible any more, but we can
keep this definition by taking the pseudo inverse 
and the theorems below will still hold true
(we refer to \cite{collins-sniady} for historical references and further details). We shall use the shorthand notation $\Wg(\sigma) = \Wg(n, \sigma)$ when the dimension parameter $n$ is clear from context.

The function $\Wg$  is used to compute integrals with respect to 
the Haar measure on the unitary group (we shall denote by $\U(n)$ the unitary group acting on an $n$-dimensional Hilbert space). The first theorem is as follows:

\begin{theorem}
\label{thm:Wg}
 Let $n$ be a positive integer and
$i=(i_1,\ldots ,i_p)$, $i'=(i'_1,\ldots ,i'_p)$,
$j=(j_1,\ldots ,j_p)$, $j'=(j'_1,\ldots ,j'_p)$
be $p$-tuples of positive integers from $\{1, 2, \ldots, n\}$. Then
\begin{multline}
\label{bid} \int_{\U(n)}U_{i_1j_1} \cdots U_{i_pj_p}
\overline{U_{i'_1j'_1}} \cdots
\overline{U_{i'_pj'_p}}\ dU=\\
\sum_{\sigma, \tau\in \S_{p}}\delta_{i_1i'_{\sigma (1)}}\ldots
\delta_{i_p i'_{\sigma (p)}}\delta_{j_1j'_{\tau (1)}}\ldots
\delta_{j_p j'_{\tau (p)}} \Wg (n,\tau\sigma^{-1}).
\end{multline}

If $p\neq p'$ then
\begin{equation} \label{eq:Wg_diff} \int_{\U(n)}U_{i_{1}j_{1}} \cdots
U_{i_{p}j_{p}} \overline{U_{i'_{1}j'_{1}}} \cdots
\overline{U_{i'_{p'}j'_{p'}}}\ dU= 0.
\end{equation}
\end{theorem}

Since we shall perform integration over \emph{large} unitary groups, we are interested in the values of the Weingarten function in the limit $n \to \iy$. The following result encloses all the data we need for our computations
about the asymptotics of the $\Wg$ function; see \cite{collins-imrn} for a proof.

\begin{theorem}\label{thm:mob} For a permutation $\sigma \in \S_p$, let $\text{Cycles}(\sigma)$ denote the set of cycles of $\sigma$. Then
\begin{equation}
\Wg (n,\sigma)=\prod_{c\in \text{Cycles} (\sigma )}\Wg (n,c)(1+O(n^{-2}))
\end{equation}
and 
\begin{equation}
\Wg (n,(1,\ldots ,d) ) = (-1)^{d-1}c_{d-1}\prod_{-d+1\leq j \leq d-1}(n-j)^{-1}
\end{equation}
where $c_i=\frac{(2i)!}{(i+1)! \, i!}$ is the $i$-th Catalan number.
\end{theorem}

As a shorthand for the
quantities in Theorem \ref{thm:mob}, we introduce the function $\Mob$ on the symmetric
group. $\Mob$ is invariant under conjugation and multiplicative over the cycles; further, it satisfies
for any permutation $\sigma \in \S_p$:
\begin{equation}\label{weingartenapprox}
\Wg(n,\sigma) = n^{-(p + |\sigma|)} (\Mob(\sigma) + O(n^{-2}))
\end{equation}
where $|\sigma |=p-\# \sigma $ is the \emph{length} of $\sigma$, i.e. the minimal number of transpositions that multiply to $\sigma$. We refer to \cite{collins-sniady} for details about the function $\Mob$ but
what we have to know in this paper is that 
\be
\Mob (\sigma) \sim (-1)^{|\sigma|}
\ee
when $\sigma$ consists of disjoint transpositions.  

We finish this section by a well known lemma which we will use several times towards the end of the paper. This result is contained in \cite{nica-speicher}.
\begin{lemma}\label{lem:S_p}
The function
$d(\sigma,\tau) = |\sigma^{-1} \tau|$ is an integer valued distance on $\S_p$. Besides, it has the following properties:
\begin{itemize}
\item the diameter of $\S_p$ is $p-1$;
\item $d(\cdot, \cdot)$ is left and right translation invariant;
\item for three permutations $\sigma_1,\sigma_2, \tau \in \S_p$, the quantity $d(\tau,\sigma_1)+d(\tau,\sigma_2)$
has the same parity as $d(\sigma_1,\sigma_2)$;
\item the set of geodesic points between the identity permutation $\id$ and some permutation $\sigma \in \S_p$ is in bijection with the set of non-crossing partitions smaller than $\pi$, where the partition $\pi$ encodes the cycle structure of $\sigma$. Moreover, the preceding bijection preserves the lattice structure. 
\end{itemize}
\end{lemma}

\subsection{Graphical calculus}

We recall briefly in this subsection the graphical calculus method for computing unitary integrals introduced in \cite{cn1}. For more details on this method,
we refer the reader to the paper \cite{cn1} and to other work which make use of this technique \cite{cn3,cn-entropy,collins-nechita-zyczkowski}

In the graphical calculus matrices (or, more generally, tensors) are represented by boxes.
Each box has differently shaped symbols, 
where the number of different types of them equals that of different spaces (exceptions are mentioned bellow). 
Those symbols are empty (while) or filled (black), corresponding to primal or dual spaces. 
Wires connect these symbols, corresponding to tensor contractions. A diagram is a collection of such boxes and wires and corresponds to an element of an abstract element in a tensor product space.

The main advantage of such a representation is that it provides an efficient way of computing expectation values of such tensors when some (or all) of the boxes are random unitary matrices. 
into an efficient way to implement the Weingarten calculus:
the delta functions in each summand in the RHS of (\ref{bid}) 
describes how we connect boxes. Each pair of permutations $(\sigma,\tau)$ in  (\ref{bid})
eliminate $U$ and $\overline U$ boxes and 
reconnect the wires originally connected to these boxes
to get a new graph. 

This process for a fixed permutation is called a removal and
the whole process which sums all the new graphs over the all permutations is called the graph expansion. 
Importantly, the graphical calculus works linearly and
separated components are related by tensor product,
as is assumed implicitly above.   
In this setting, the Bell state is represented by
a wire connecting two black symbols. 

We are allowed to make tensor-product of some spaces
and decompose it into different spaces 
giving them different symbols from the original set of symbols. 
In this case, we must bear in mind that
these two set of symbols represent the same product space
when we expand the graph.

\section{Generalized Bell states for $U\times \bar U$.}\label{sec:generalized-Bell}

In the seminal paper of Hastings \cite{hastings}, violation of additivity was found when the dimensions 
of input and output spaces are much larger than that of the environment space of the channel. 
We also follow this scenario, by setting $d_{out} = n$, $d_{env} = k$ and $d_{in} = tnk$.
Here, the integer $k$ and $t \in (0,1)$ are fixed parameters and we let $n$ go infinity to observe the asymptotic behavior.  
Since the output space ($\mathbb C^n$) is larger than the environment space ($\mathbb C^k$),
we investigate the 
complementary
channels so that we have to study output matrices of smaller, fixed dimensions. As it was discussed in Section \ref{sec:review}, the value of the minimum output entropy does not change, since the eigenvalues of the two partial traces of a rank one projector are the same, up to zeroes. 

We are thus interested in the following $k \times k$ random matrix:
\be
Z_n = \Phi^C \otimes \overline{\Phi}^C (|\psi_n\ket  \bra \psi_n| ) 
\ee
Here, $|\psi_n\ket$ is a fixed input vector for each $n$ and
$\Phi$ is random in the measure defined above. 

To represent the input $|\psi\ket$ in the graphical calculus
we add $A$ and $A^*$ boxes on the wires of the Bell input (See Figure \ref{AAstar}). 
\begin{figure}[htbp]
 \begin{center}
  \includegraphics[]{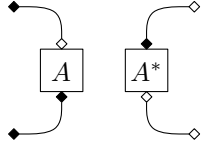}
 \end{center}
 \caption{Generalized Bell state} 
 \label{AAstar}
\end{figure}

Algebraically, we consider a sequence of inputs
\be\label{input-formula}
|\psi_n\rangle = \sum_{i,j =1}^{tnk} a_{ij} |i\rangle| j\rangle
\qquad a_{ij} \in \C
\ee
Then 
\be
A_n = \sum_{i,j =1}^{tnk} a_{ij} |i\rangle \langle j|
\ee
is used in the graphical calculus. 
In order to ensure normalized vectors, we choose $A_n \in \mathbb M_{tnk}(\mathbb C)$ to be such that
\be 
\trace [A_n A_n^*] =1
\ee

We consider first cases when $A$ is scaled properly
as a generalization of the Bell state input,
when interesting phenomena occur. 
Some of other cases will be treated later in this section.

\subsection{Well-behaved input}
In order to define well-behaved inputs, we introduce two assumptions on the asymptotic behavior of the sequence of input states $A_n$.
\\
{\bf Assumption 1:}
\be
\frac{\trace \left[ A_n \right] }{\sqrt{tnk}} = m + O\l(\frac{1}{n^2}\r)
\ee
for some $m\in\C$.
Note that the similar limit for $A_n^*$ is $\bar m$.
\\
{\bf Assumption 2:} 
\begin{equation}
\|A_n\|_\infty = O\l(\frac{1}{\sqrt{n}}\r)
\end{equation}
  
With these two assumptions, we can prove the following result. Recall that the empirical eigenvalue distribution of a matrix $Z \in \mathbb M_{k^2}(\mathbb C)$ is the probability measure
$$k^{-2} \sum_{i=1}^{k^2} \delta_{\lambda_i},$$
where $\lambda_1, \ldots, \lambda_{k^2}$ are the eigenvalues of $Z$.
\begin{theorem}
\label{thm:bell-phenomenon}
Under Assumptions 1 and 2, the empirical eigenvalue distribution of the matrix $Z_n$ converges \emph{almost surely}, as $n \rightarrow \infty$, to the probability measure
\be\label{limitdistUbar}
\frac{1}{k^2}\l[\delta_{\lambda_1} + (k^2-1) \delta_{\lambda_2}\r]
\ee
where the Dirac masses are located at
\be
\lambda_1 = t |m|^2 + \frac{1-t |m|^2}{k^2}
\quad \text{and} \quad
\quad \lambda_2 = \frac{1-t |m|^2}{k^2}.
\ee
In other words, the output state has asymptotically the following eigenvalues:
\begin{itemize}
\item $t |m|^2 + \frac{1-t |m|^2}{k^2}$, with multiplicity one;
\item $\frac{1-t |m|^2}{k^2}$, with multiplicity $k^2-1$.
\end{itemize}
\end{theorem}
\begin{proof}
The proof uses the moment method and consists of two steps. First, we compute the asymptotic moments of the output density matrix $Z$ and then, by a Borel-Cantelli argument, we deduce the almost sure convergence of the spectral distribution and of the eigenvalues.

{\bf Step 1:} We calculate the limit moments of $Z_n$, using the graphical calculus, see Figure \ref{UUbar-AAstar}.
Here,
$\includegraphics{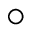}$ and $\includegraphics{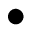}$ correspond to the $n$-dimensional environment space of $\Phi^C$, and
$\includegraphics{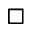}$ and $\includegraphics{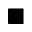}$ to the $k$-dimensional output space,
whereas $\includegraphics{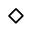}$ and $\includegraphics{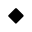}$ 
correspond to the $tnk$-dimensional input space. 
\begin{figure}[htbp]
 \begin{center}
  \includegraphics[]{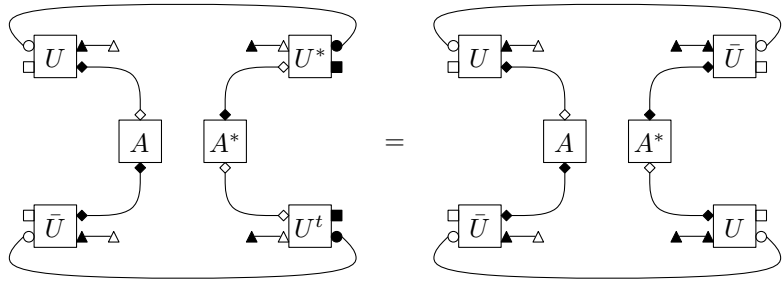}
 \end{center}
 \caption{$\Phi \otimes \bar\Phi$ with generalized input} 
 \label{UUbar-AAstar}
\end{figure}  

For fixed $p \in \N$ the Weingarten sums are indexed by pairs of permutations $(\alpha, \beta) \in \S_{2p}^2$. 
Here, we label the $U$ and the $\ol U$ boxes in the following manner: 
$1^T, 2^T, \ldots, p^T$ for the $U$ boxes of the first channel (T as ``top'') and 
$1^B, 2^B, \ldots, p^B$ for the $U$ boxes of the second channel (B as ``bottom''). 
We shall also order the labels as $\{1^T, 2^T, \ldots, p^T, 1^B, 2^B, \ldots, p^B\} \isom \{1, \ldots, 2p\}$. 
A removal $r=(\alpha, \beta) \in \S_{2p}^2$ of the random ($U$ and $\ol U$)  boxes connects the decorations in the following way:
\begin{enumerate}
	\item the white decorations of the $i$-th $U$-block are paired with the white decorations of the $\alpha(i)$-th $\ol U$ block;
	\item the black decorations of the $i$-th $U$-block are paired with the black decorations of the $\beta(i)$-th $\ol U$ block.
\end{enumerate}
Next, we introduce two fixed permutations $\gamma, \delta \in \S_{2p}$ which represent wires appearing in the diagram before the graph expansion.
The permutation $\gamma$ represents the initial wiring of the $\includegraphics{square_w.eps}$ decorations
and $\delta$ accounts for the wires between the $\includegraphics{diamond-black.eps}$ decorations connecting boxes $A$ or $A^*$.
More precisely, for all $i$,
\begin{equation} 
\label{eq:def-gamma-delta}
\gamma(i^T) = (i-1)^T, \quad \gamma(i^B) = (i+1)^B,\quad
and \quad
\delta(i^T) = i^B, \quad \delta(i^B) = i^T.
\end{equation}
A difference between this calculation and the one in \cite{cn1} is that 
the Bell state with $\includegraphics{circle_b.eps}$ in that paper turned into 
$\includegraphics{diamond-black.eps}$ with the A boxes here.
So, for these wires we get ``necklaces'' with $A$ or $A^*$ beads instead of just loops. 
Hence, we now can list of calculation elements in the graphical calculus for each $(\alpha,\beta) \in S_{2p}\times S_{2p}$:
\begin{enumerate}
	\item ``$\includegraphics{square_w.eps}$''-loops: $k^{\#(\gamma^{-1}\alpha )}$;
	\item ``$\includegraphics{circle_w.eps}$''-loops: $n^{\# \alpha}$;	
	\item ``$\includegraphics{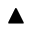}$''-loops: none;
	\item ``$\includegraphics{diamond-black.eps}$''-``necklaces'':$f(\beta)$, which is to be defined below;
	\item Weingarten weights for the $U$-matrices: $\Wg(\alpha\beta^{-1})=\Wg (\alpha^{-1} \beta)$.
\end{enumerate}
Therefore, the Weingarten graph expansion formula (Theorem 4.1 of \cite{cn1}) reads
\be \label{moments}
\E \trace [Z_n^p] = \sum_{\alpha,\beta \in S_{2p}} n^{\#\alpha} k^{\#(\gamma^{-1} \alpha)} f(\beta) \Wg (\alpha^{-1} \beta)
\ee
Here,
\be
f(\beta) &=& \prod_{c \in {\rm Cycle}(\beta^{-1}\delta)} 
\trace \left[ A^{s_{c,1}} \cdots A^{s_{c,\|c\|}}  \right]
\ee
Here, $\|c\|$ is the number of elements in $c$ and 
$s_{c,1} \ldots ,s_{c,\|c\|}$ are defined such that
\be 
s_{c,i}=
\begin{cases}
1 & \text{if the $i$th element in the cycle $c$ belongs to $T$}\\
* & \text{if the $i$th element in the cycle $c$ belongs to $B$}
\end{cases}
\ee
Note that the above function $f(\beta)$ is well-defined 
in spite of the ambiguity of $s_{c,i}$, because of the circular property of the trace. 
It is also well-behaved, in fact,
since for each $c$
\be
&&|\trace \left[ A^{s_{c,1}} \cdots A^{s_{c,\|c\|}}  \right]| 
\leq 
\| A^{s_{c,1}} \|_\infty \cdots \| A^{s_{c,\|c\| -1}}\|_\infty \cdot \| A^{s_{c,\|c\|}} \|_1 \\ \notag
&&\lesssim
\l(\frac{1}{\sqrt{n}} \r)^{\|c\|-1} \cdot\sqrt{n}
= n^{1-\|c\|/2} 
\ee
we get the following bound
\be\label{roughbound_f}
|f(\beta)|  \lesssim   n^{\#(\beta^{-1}\delta) -p } 
\ee

Hence,
\bee
 \E \trace [Z_n^p] 
\lesssim \sum_{\alpha,\beta \in S_{2p}} n^{\#\alpha} 
n^{\#(\beta^{-1}\delta) -p} 
n^{-2p - |\alpha^{-1}\beta|} 
\qquad \text{as $n \rightarrow \infty$}
\label{bound_f}
\eee
The power of $n$ in the RHS of (\ref{bound_f}) as a whole is  
\bee
2p- |\alpha| +p-|\beta^{-1}\delta| -2p - |\alpha^{-1}\beta| 
= p- (|\alpha| + |\alpha^{-1}\beta| + |\beta^{-1}\delta|) \leq 0
\eee
Here, equality holds if and only if 
$\id \rightarrow \alpha \rightarrow \beta \rightarrow \delta$ is a geodesic:
\be\label{geodesic_ab}
\alpha = \prod_{i\in A} \tau_i, \qquad \beta = \prod_{i\in B} \tau_i
\ee
where $\tau_i = (i^T,i^B)$ and $A\subseteq B\subseteq \{1, \ldots,p\}$; we refer to
\cite{cn1} for the proof.
Importantly, for these $\beta$ we have the following asymptotic behavior:
\bee\label{behaviour_f}
f (\beta) =  \l(tnk|m|^2\r)^{|\beta|}  + O \l(\frac{1}{n^2}\r)  
\eee
Note that $|\beta|=|B|$.
This implies that the power of $n$  in (\ref{moments}) in fact becomes $0$
for all the $\alpha,\beta$ which satisfy the geodesics condition  
$\id \rightarrow \alpha \rightarrow \beta \rightarrow \delta$:
\be
\#\alpha +  |B| -2p - |\alpha^{-1}\beta| =
2p - |A| + |B|  -2p - |B \setminus A| =0
\ee 
Here,  $|B| = |B \setminus A|+ |A|$. Next, we make use of the following lemma.
\begin{lemma}\label{countingloops}
For a geodesic: $\id \rightarrow \alpha \rightarrow \beta \rightarrow \delta$
we have
\be
\#(\gamma^{-1} \alpha) &=& \begin{cases} 2 & A = \emptyset \\ |A| & A \not= \emptyset \end{cases} \\
\#(\delta \beta) &=& p+|B| \\
|\alpha^{-1} \beta| &=& |B\setminus A|
\ee
\end{lemma}
We thus have the following approximation on $\E \trace [Z_n^p] $ (note that the estimate on the error order is not necessary here but will be so in Step 2) :
\begin{align}\label{momentsapprox}
& \Ex  \l[Z_n^p\r] 
= \sum_{\id \rightarrow \alpha \rightarrow \beta \rightarrow \delta}
\l( n^{\#\alpha} k^{\#(\gamma^{-1} \alpha)} f(\beta) + O\l( \frac{1}{n^2}\r) \r)\Wg (\alpha^{-1} \beta)   \\
&=  \sum_{\id \rightarrow \alpha \rightarrow \beta \rightarrow \delta}
 k^{\#(\gamma^{-1} \alpha)}\l(tk|m|^2\r)^{|\beta|} 
k^{-2p-|\alpha^{-1} \beta|} \Mob (\alpha^{-1} \beta) + O\l( \frac{1}{n^2}\r) \notag
\end{align}
Note that $|\alpha^{-1}\beta|=|B\setminus A|$.
The first equality holds because Lemma \ref{lem:S_p} implies that 
permutations $(\alpha,\beta)\in S_{2p} \times S_{2p}$ off the geodesic 
make the power of $n$ less by two or more; only even powers are allowed. 
The second inequality comes from (\ref{weingartenapprox}) and (\ref{behaviour_f}).

Hence,
\be
\lim_{n\rightarrow \infty}  \E \trace [Z_n^p]
&=& 
\underbrace{\sum_{A\subseteq  B} 
k^{|A|}\l(tk|m|^2\r)^{|B|} k^{-2p - |B\setminus A|} (-1)^{|B\setminus A|} }_{(a)} \notag\\
&-& \underbrace{\sum_{\emptyset=A \subseteq B} 
k^{|A|}\l(tk|m|^2\r)^{|B|} k^{-2p - |B\setminus A|} (-1)^{|B\setminus A|} }_{(b)} \notag\\
&+&\underbrace{\sum_{\emptyset =A \subseteq B} 
k^{2}\l(tk|m|^2\r)^{|B|} k^{-2p - |B\setminus A|} (-1)^{|B\setminus A|}  }_{(c)} \notag
\ee
The combinatorial sums above can be computed using the following multinomial identities.
\begin{lemma}\label{multinomialidentities}
\begin{align*}
\sum_{\emptyset \subseteq A \subseteq \{1, \ldots, p\}} \!\!\!\! x^{|A|} &= (1+x)^p \quad \text{and}\\
\sum_{\emptyset \subseteq A \subseteq B \subseteq \{1, \ldots, p\}} \!\!\!\!\!\!\!\! x^{|A|}y^{|B \setminus A|} &= (1+x+y)^p.
\end{align*}
\end{lemma}
We obtain that
\be
(a) &=& k^{-2p} \sum_{A\subseteq  B} \l(k \cdot tk \cdot |m|^2\r)^{|A|} 
\times \l(tk\cdot |m|^2\cdot k^{-1} \cdot (-1) \r)^{|B\setminus A|} \\
&=& k^{-2p} (1 + tk^2  |m|^2 -t |m|^2 )^p \notag.
\ee
Also,
\be
(b) = k^{-2p}\sum_{B} \l(tk \cdot |m|^2 \cdot k^{-1} \cdot (-1)\r)^{|B|} 
= k^{-2p} \l(1- t |m|^2\r)^p
\ee
and similarly,
\be
(c) = k^{2-2p}\sum_{B} \l(tk \cdot |m|^2\cdot k^{-1} \cdot (-1) \r)^{|B|} 
= k^{2-2p} \l(1- t |m|^2\r)^p.
\ee
Therefore,
\bee
\lim_{n\rightarrow \infty}  \E \trace [Z_n^p]
= \left[ \frac{1}{k^2} + \frac{(k^2-1) t|m|^2}{k^2}\right]^p 
+ (k^2-1) \left[ \frac{1}{k^2} - \frac{ t|m|^2}{k^2} \right]^p,
\label{limitmomentUbar}
\eee
which completes the proof of the first step.

\bigskip

{\bf Step 2:} 
We now move on to prove the almost sure convergence.
Since this part of proof is very similar to that of Theorem 6.3 in \cite{cn1}
we only show the sketch of proof here. 
Via the Borel-Cantelli Lemma, all we have to prove is that
the covariance series converges:
\bee\label{series}
\sum_{n=1}^\infty \Ex \l[
\l(\trace \l[ Z_n^p\r] - \Ex \trace \l[ Z_n^p \r] \r)^2\r]
= \sum_{n=1}^\infty \Ex \l[
\l(\trace \l[ Z_n^p\r] \r)^2 \r]- \l(\Ex \trace \l[ Z_n^p \r] \r)^2
 < \infty
\eee
which implies that for all $p\geq 1$  
\be
\trace \l[ Z_n^p\r]   \rightarrow (\ref{limitmomentUbar}) 
\quad \text{a.e.} \quad \text{as $n \rightarrow \infty$}
\ee 
Indeed, this shows the convergence of empirical distribution.
Also, by Carleman's condition, equation
(\ref{limitmomentUbar}) uniquely determines the measure as in (\ref{limitdistUbar}).

First, (\ref{momentsapprox}) implies that 
\begin{align}\label{moment^2}
\l(\Ex \trace \l[ Z_n^p \r] \r)^2
=  
&\l(\sum_{\id \rightarrow \alpha \rightarrow \beta \rightarrow \delta}
 k^{\#(\gamma^{-1} \alpha)}\l(tk|m|^2\r)^{|\beta|} 
k^{-2p-|\alpha^{-1} \beta|} \Mob (\alpha^{-1} \beta)\r)^2 \\
&+ O\l( \frac{1}{n^2}\r).\notag
\end{align}

We then calculate $\Ex [(\trace \l[ Z_n^p] \r)^2 ]$.
In the diagram we have two identical copies of $\trace Z_n^p$, which means that
we have $4p$ pairs of $U$ and $\overline U$ boxes.
As a result, 
removals $(\bar\alpha,\bar\beta)$ are defined for $\bar\alpha,\bar\beta \in S_{4p}$. 
However, importantly those two copies are initially separated. 
Namely, initial wires $\bar \gamma, \bar \delta \in S_{2p}\oplus S_{2p} = S_{4p}$ 
are written as direct sums:
\be
\bar \gamma = \gamma \oplus \gamma 
\quad \text{and} \quad 
\bar\delta = \delta \oplus \delta
\ee
Then, as before, we calculate the power of $n$, which is 
\be
2p- (|\bar\alpha| + |\bar\alpha^{-1}\bar\beta| + |\bar\beta^{-1}\bar\delta|) \leq 0
\ee
Here, ``$=$'' holds if and only if 
$\id \rightarrow \bar\alpha \rightarrow \bar\beta \rightarrow \bar\delta$
is a geodesic. 
Moreover, this geodesic condition implies that
$\bar\alpha$ and $\bar\beta$ can be written as
\be
\bar\alpha = \alpha_1 \oplus \alpha_2 
\quad \text{and} \quad
\bar\beta =   \beta_1 \oplus \beta_2
\ee 
Here, 
pairs $(\alpha_1, \beta_1)$ and $(\alpha_2, \beta_2)$ are 
defined as in (\ref{geodesic_ab}).

Therefore, in the diagram all removals which matter as $n\rightarrow \infty$
keep those two copies separated. 
Also, these removals have the following properties:
\begin{align*}
&\# (\bar\gamma^{-1} \bar \alpha) 
= \# (\gamma^{-1}  \alpha_1) + \# (\gamma^{-1}  \alpha_2), \\
&|\bar\beta| = |\beta_1|+|\beta_2|, \qquad
|\bar\alpha^{-1} \bar\beta |= |\alpha_1^{-1} \beta_1 | + |\alpha_2^{-1} \beta_2 |   \\
&\Mob (\bar\alpha^{-1} \bar\beta ) = \Mob (\alpha_1^{-1} \beta_1 ) \cdot \Mob(\alpha_2^{-1} \beta_2 )
\end{align*}
As before, 
we get an approximation with the error of order $1/n^2$:
\begin{align}\label{2ndmoment}
&\Ex  \l(\trace\l[ Z_n^p\r] \r)^2 
=  
\sum_{\substack{\id \rightarrow \alpha_1 \rightarrow \beta_1 \rightarrow \delta \\
\id \rightarrow \alpha_2 \rightarrow \beta_2 \rightarrow \delta}}
\Big[
 k^{\#(\gamma^{-1} \alpha_1)+\#(\gamma^{-1} \alpha_2)}
\l(tk|m|^2\r)^{|\beta_1|+|\beta_2|} \\
& \qquad k^{-4p-|\alpha_1^{-1} \beta_1|+|\alpha_2^{-1} \beta_2|} 
\Mob (\alpha_1^{-1} \beta_1)\cdot \Mob (\alpha_2^{-1} \beta_2) \Big]
+ O\l( \frac{1}{n^2}\r)\notag
\end{align}
Here, the error is of order $1/n^2$ for the same reason as before. 

Finally, we see from (\ref{moment^2}) and(\ref{2ndmoment}) that
\bee
\Ex \l[\l(\trace \l[ Z_n^p\r] \r)^2 \r]- \l(\Ex \trace \l[ Z_n^p \r] \r)^2
 = O \l( \frac{1}{n^2}\r)
\eee
which proves (\ref{series}), and finalizes the proof.
\end{proof}

\subsection{Consequence of Theorem \ref{thm:bell-phenomenon}}
We start by considering some special cases of interest where the previous theorem applies. 
\subsubsection*{Example 1: the Bell state.}

The original, non-perturbed Bell state  corresponds to the following matrix $A$:
\be
A = \frac{1}{\sqrt{tnk}} \cdot {\rm diag} \{1, \ldots, 1\}
\ee
Then,
\be
m = \lim_{n \rightarrow \infty} tnk \cdot \frac{1}{tnk} =1
\ee
Hence, we have the following limit eigenvalue distribution: 
\begin{itemize}
\item $t  + \frac{1-t }{k^2}$, with multiplicity one;
\item $\frac{1-t}{k^2}$, with multiplicity $k^2-1$.
\end{itemize}
In particular, when $t= 1/k$ (the dimension of input space is $n$), we get
\begin{itemize}
\item $\frac{1}{k}  + \frac{1}{k^2} - \frac{1}{k^3}$, with multiplicity one;
\item $\frac{1}{k^2}- \frac{1}{k^3}$, with multiplicity $k^2-1$,
\end{itemize}
recovering in this way the results of \cite{cn1}. 

Note that this method yields better bounds for the largest eigenvalue (and thus for the entropy) than a direct application of the Hayden-Winter trick, which only gives a crude bound ($\lambda_1 \geq t$) on the largest eigenvalue of the output. 

\subsubsection*{Example 2: Dephased Bell state.}
Take the input to be the maximally entangled state with phases as bellow.
\be
|\phi\ket 
= \frac{1}{\sqrt{tnk}}\,
 \sum_{j =1}^{tnk} \exp \l\{ \frac{2\pi i}{tnk} j \r\}|j\ket|j\ket
\ee
where $i^2 =-1$. 
The corresponding matrix $A$ is
\be
A = \frac{1}{\sqrt{tnk}} \,
{\rm diag} \l\{\exp \l\{ \frac{2\pi i}{tnk}  \r\},\exp \l\{ \frac{4 \pi i}{tnk}  \r\},
\ldots, 1  \r\}
\ee
so that $m=0$.
Hence, we have the following flat limit eigenvalue distribution: $k^2$ eigenvalues equal to $k^{-2}$.

\subsubsection*{Example 3: Normal matrices with asymptotic moments.}

We consider now a generalization of the previous two situations, where the matrices $A_n$ are normal and have asymptotic moments in the sense that the following limit exists:

\begin{equation}
\forall p,q \geq 0, \quad
m_{p, q} 
=	\lim_{n \to \iy} \frac{\trace \left[ A_n^p (A_n^*)^q \right] }{(tnk)^{1 - (p+q)/2} }
= \int z^p \bar{z}^q \d{\mu}(z),
\end{equation}
where $\mu$ is a compactly supported measure on the complex plane. Since we want input states to be normalized, 
we assume $m_{1, 1} = 1$.  Obviously, such matrices satisfy Assumption 1 of Theorem \ref{thm:bell-phenomenon}, with $m=m_{1,0}$. Although Assumption 2 may not be satisfied, note that one can still get the bound in equation \eqref{roughbound_f} for the $*$-moments of $A_n$ using normality and the hypothesis above; since this is the only place in the proof of Theorem \ref{thm:bell-phenomenon} when one uses Assumption 2, the result holds. 

\bigskip

The next theorem, one of the main results of the paper, is an easy consequence of Theorem \ref{thm:bell-phenomenon} and presents the usual Bell state as the unique input state 
of a conjugate product channel that yields an output with minimal entropy, within the class of well-behaved input states.

\begin{theorem}\label{thm:Bell-optimal}
Among all generalized Bell states described by a sequence of matrices $A_n$ satisfying the assumptions of Theorem \ref{thm:bell-phenomenon}, the one that achieves 
an output with minimal entropy is the usual Bell state, obtained by setting $A_n = \mathrm{I}_{tnk}/\sqrt{tnk}$.
\end{theorem}
\begin{proof}
At fixed $t$, the largest eigenvalue of the limiting measure is an increasing function of $|m|$. Hence, the output with te least entropy is obtained for the larges achievable 
value of $|m|$. The normalization constraint $\mathrm{Tr}[A_nA_n^*] = 1$ implies that one must have $|m|\leq 1$, with equality if and only if $A_n$ is equal, up to a phase, to the identity matrix $\mathrm{I}_{tnk}/\sqrt{tnk}$.
\end{proof}

\begin{remark}
As it was noted in the introduction, the previous result does not imply that the Bell state achieves the minimum in the formula for the minimal output entropy of a product of conjugate channels. 
Our result just states that the Bell state achieves the least entropy when compared to other input states which exhibit a nice eigenvalue behaviour, in the sense of the two assumptions 
appearing before Theorem \ref{thm:bell-phenomenon}. The global minimum for the output entropy could in principle be achieved on a state not having such a behaviour. 
\end{remark}

\subsection{Ill-behaved input}
We examine in this subsection input states which do not satisfy Assumption 2. 
\\
{\bf Case 1:}
When the fixed input is a product state,
i.e., the corresponding matrix $A$ is of rank one,
we get the eigenvalue $1/k^2$ with multiplicity $k^2$. 
Since the input is a product state,
outputs are products of outputs of those two channels.
Hence the above conclusion is derived from \cite{cn3}.
\\
{\bf Case 2:}
Instead of Assumptions 1 and 2 of the previous section, we set the following conditions on the matrix $A$:
\be
{\rm rank} (A) \lesssim n^{1 -\epsilon}
\quad \text{and} \quad 
\| A \|_\infty \lesssim n^{\frac{\epsilon -1}{2}} 
\ee
Note that the first condition is equivalent to
taking the dimension of input space to be of oder $n^{1-\epsilon}$.
The second condition is similar to Assumption 2, preventing the input state from having large Schmidt coefficients.

\begin{proposition}
Under these two assumptions
the empirical eigenvalue distribution of the matrix $Z_n$ converges \emph{almost surely}, as $n \rightarrow \infty$, to the probability measure $\delta_{k^{-2}}$.
\end{proposition}
\begin{proof}
The expansion of graph works in the exactly same way 
as in the proof of Theorem \ref{thm:bell-phenomenon}.
A difference is that we get a different (rough) upper-bound for $f(\beta)$ from (\ref{roughbound_f}):
\be
|f(\beta)|  \lesssim   n^{(1-\epsilon)(\#(\beta^{-1}\delta) -p)} 
\ee
which implies 
\bee
 \E \trace [Z_n^p] 
\lesssim \sum_{\alpha,\beta \in S_{2p}} n^{\#\alpha} 
n^{(1-\epsilon)(\#(\beta^{-1}\delta) -p)} 
n^{-2p - |\alpha^{-1}\beta|} 
\qquad \text{as $n \rightarrow \infty$}.
\eee
The power of $n$ in the general term of the sum above is
\be
&& p- (|\alpha| + |\alpha^{-1}\beta| + |\beta^{-1}\delta|) - \epsilon \l(p-|\beta^{-1}\delta|\r) \label{irrepower1}\\
 &=&  - |\alpha| -  |\alpha^{-1}\beta| + (1- \epsilon)\l(p- |\beta^{-1}\delta|\r) \label{irrepower2}
\ee
which is non-positive.
Indeed, it is obvious from (\ref{irrepower1}) when $p-|\beta^{-1}\delta|\geq 0$ and 
from (\ref{irrepower2}) otherwise. 
To achieve equality, we need to have $\alpha = \beta= \id$, at least.
In this case, $f(\id) = 1$ and the power of $n$ in (\ref{moments}) in fact becomes $0$:
\be
\#(\alpha) - 2p - |\alpha^{-1}\beta| = 2p - 2p -0 =0
\ee
Therefore, 
\be
\lim_{n\rightarrow \infty}  \E \trace [Z_n^p]
= k^{\#(\gamma^{-1}\alpha)} \cdot k^{-2p-|\alpha^{-1}\beta|} \Big|_{\alpha=\beta=\id} 
= k^2 \cdot (k^{-2})^p,
\ee
and we recover the announced flat limiting distribution.
\end{proof}

As a final remark, note that in this case, the Hayden-Winter trick
\be
\text{ (RHS of (\ref{HWtrick}))} \sim \frac{n^{1-\epsilon}}{n} \rightarrow 0 
\qquad \text{as $n \rightarrow \infty$} 
\ee
does not produce a useful bound for the larges eigenvalue of the output.

\section{Generalized Bell states for $U\times U$}\label{sec:generalized-Bell-UU}
We investigate in this section tensor products of two identical random quantum channels. 
Again, an input is fixed and a channel is drawn randomly 
according to
the Haar measure. 
More precisely, we are interested in the following $k \times k$ random matrix:
\be
Z_n = \Phi^C \otimes \Phi^C (|\psi_n\ket  \bra \psi_n| ).
\ee
As in the previous section, $|\psi_n\ket$ is a fixed input vector for each $n$.
The diagram of $Z_n$ is presented in Figure \ref{UU-AAstar}. 
We stud well-behaved inputs and we show that the output $Z_n$ has in this case an asymptotically flat eigenvalue distribution. 
The graph expansion method works in the same way except for
small modifications on $\gamma$ and ``necklaces''.
\begin{figure}[htbp]
 \begin{center}
  \includegraphics[]{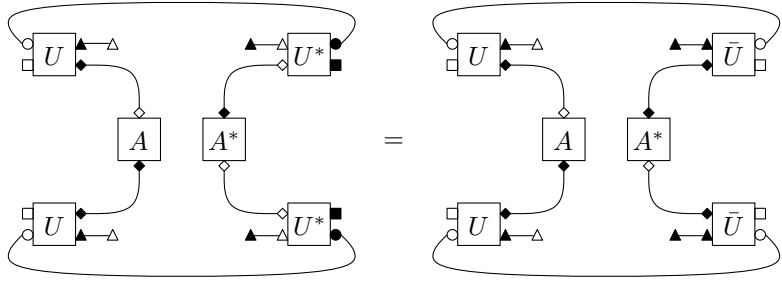}
 \end{center}
 \caption{Diagram for the output of $\Phi \otimes \Phi$ with generalized input} 
 \label{UU-AAstar}
\end{figure}  

The implications of the following theorem and a discussion on the reasons for which conjugate channels and identical channels give different results will be done in the next section. 

\begin{theorem}
\label{thm:bell-phenomenon2}
Under those two assumptions in Theorem \ref{thm:bell-phenomenon},  
the empirical eigenvalue distribution of the matrix $Z_n$ converges, as $n \rightarrow \infty$, to
$\delta_{1/k^2}$.
Hence, $Z_n$ has asymptotically a flat spectrum.
\end{theorem}
\begin{proof}
We introduce a fixed permutations $\tilde\gamma \in \S_{2p}$ which  
represents the initial wiring of the $\includegraphics{square_w.eps}$ decorations.
More precisely, for all $i$,
\begin{equation} 
\label{eq:def-tildegamma}
\tilde\gamma(i^T) = (i+1)^T, \quad \tilde\gamma(i^B) = (i+1)^B
\end{equation}
The following is the list of calculation elements in the graphical calculus 
for each $(\alpha,\beta) \in S_{2p}\times S_{2p}$:
\begin{enumerate} 
	\item ``$\includegraphics{square_w.eps}$''-loops: $k^{\#(\tilde\gamma\alpha )}$;
	\item ``$\includegraphics{circle_w.eps}$''-loops: $n^{\# \alpha}$;
	\item ``$\includegraphics{triangle-black.eps}$''-loops: none; 
	\item ``$\includegraphics{diamond-black.eps}$''-``necklaces'':$g(\beta)$, which is to be defined below;
	\item Weingarten weights for the $U$-matrices: $\Wg(\alpha\beta^{-1})= \Wg(\alpha^{-1} \beta) $.
\end{enumerate}
We calculate  
\be 
\E [ \trace Z_n^p] = \sum_{\alpha, \beta \in S_{2p}} n^{\# \alpha} k^{\#(\tilde\gamma\alpha)} g(\beta) \Wg(\alpha^{-1} \beta) 
\ee
The above $g(\beta)$ is defined as follows. 
Set a relation:
\be
i \sim j \Leftrightarrow \beta(i^{x}) =  k^T \quad \text{and}\quad \beta(i^{y}) =  k^B  
\quad (x,y =  T \,\text{or}\, B)
\ee
and make the smallest partition $\hat \beta$ on $\{1,\cdots,p\}$ such that any two numbers with the above relation belong to the same block.
Here, each block corresponds to a ``necklace''.
Note that each ``necklace'' carries the same number of $A$ and $A^*$
because $A$ is connected only to $U$ and $A^*$ only to $\overline U$.
Hence, $g(\beta)$ reads
\be
g(\beta) = \prod_{b \in {\rm Blocks}(\hat\beta)} 
\trace \left[A^{s_{b,1}}(A^*)^{t_{b,1}} \cdots A^{s_{b,\|b\|}} (A^*)^{t_{b,\|b\|}} 
\right]
\ee
As before, $\|b\|$ is the cardinality of the block $b$
and $s_{b,i},t_{b,i} = 1 \,\text{or}\, *$
 depending on how they are connected. 
We upper-bound $g(\beta)$ similarly as in the previous section.
Since
\be
\trace \left[A^{s_{b,1}}(A^*)^{t_{b,1}} \cdots A^{s_{b,\|b\|}} (A^*)^{t_{b,\|b\|}} \right]
\lesssim n^{1- \|b\|} 
\qquad \text{as $n \rightarrow \infty$}
\ee
we have a rough bound for $g(\beta)$:
\be
g(\beta)
\lesssim n^{\#(\hat \beta)-p } 
\qquad \text{as $n \rightarrow \infty$}
\ee
Here, $\#(\hat \beta)$ is the number of blocks in the relation induced by $\beta$.

Hence we have
\bee
\E \trace [Z_n^p] \lesssim
\sum_{\alpha, \beta \in S_{2p}} n^{\# \alpha} \cdot
n^{\#(\hat \beta)-p } \cdot
n^{-2p - |\alpha^{-1}\beta|} 
\qquad \text{as $n \rightarrow \infty$}
\eee 
The power of $n$ as a whole in the RHS is 
\be
& &2p- |\alpha|+ \#(\hat \beta) -p -2p - |\alpha^{-1}\beta|
= -p - (|\alpha| + |\alpha^{-1}\beta|) + \#(\hat\beta) \\
&\leq& -p- |\beta|+ \#(\hat\beta) \leq 0 \notag
\ee
Note that $\#(\hat\beta) \leq p$. 
It is easy to see that $=$ holds if and only if
 $\id \rightarrow \alpha \rightarrow \beta $ is a geodesic, and $|\beta|=0$ and $\#(\hat\beta)=p$.
I.e., $\alpha = \beta = \id$.
Therefore,
\be
\lim_{n \rightarrow \infty}\E \trace [Z_n^p] \sim k^{\#(\tilde \gamma)} \cdot 1 \cdot k^{-2p} 
 = k^{2-2p} = k^2 \cdot \left(\frac{1}{k^2}\right)^p,
\ee
proving the flatness of the limiting distribution.
Almost sure convergence is proven similarly as in the proof of Theorem \ref{thm:bell-phenomenon}, we leave the details to the reader.
\end{proof}

\section{What distinguishes $U\otimes \bar U$ in the graphical calculus}\label{sec:conjugated}

As we have seen in the preceding sections, conjugate channels ($U \otimes \overline U$) and identical channels 
($U \otimes U$) yield completely different behaviors. 
In this section we will explain the difference more intuitively in the graphical calculus framework
and
derive limit eigenvalue distributions in some other cases of theoretical interest:
$U\otimes U^*$ and $U \otimes U^T$. 

\subsection{Conjugate versus identical channels}
In order to compare $\Phi \otimes \bar\Phi$ and $\Phi \otimes \Phi$
we go back to Figure \ref{UUbar-AAstar} and Figure \ref{UU-AAstar} but
 $A$ will be replaced  by $\tilde A = \sqrt{n} A$ so that 
``necklaces'' are proportional to loops in $n$. 
Of course, this rescaling produces 
a multiplicative
constant $1/n$ in the graph, 
which appears as an exponent of $-p$ in (\ref{L-formula}).  

The $U \otimes U$ case is represented in Figure \ref{page2-2}.
Figure \ref{page2-2-1} shows 
the graphical calculus of $n \times \Phi^C \otimes \Phi^C (|\psi_n\ket\bra \psi_n|)$ and
we find the leading wires in Figure \ref{page2-2-2}.
Remember that we let $n$ go to the infinity. 
\begin{figure}[htbp] 
  \centering
  \subfloat[Rescaled input]{\label{page2-2-1}\includegraphics{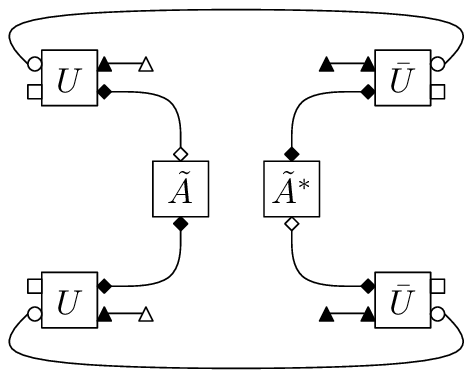}}      
\hfill     
  \subfloat[Leading wires]{\label{page2-2-2}\includegraphics{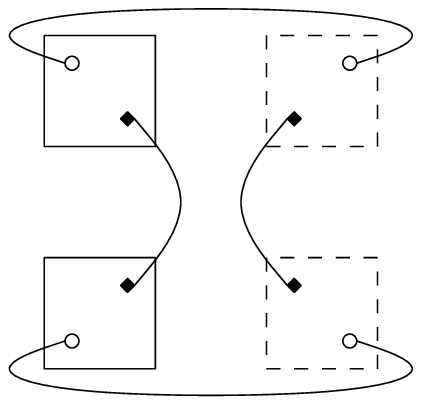}}
  \caption{Diagram for $\Phi^C \otimes \Phi^C (|\psi_n\ket\bra \psi_n|)$ 
: $U \otimes U$-case}
  \label{page2-2}
\end{figure}
The same idea is applied to $U \otimes \bar U$ 
to get the figure \ref{page2-1}.
\begin{figure}[htbp]
  \centering
  \subfloat[Rescaled input]{\label{page2-1-1}\includegraphics{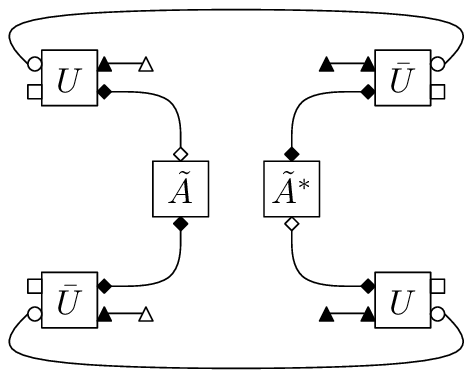}}      
\hfill          
  \subfloat[Leading wires]{\label{page2-1-2}\includegraphics{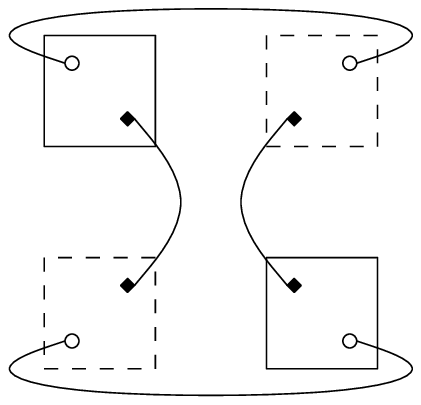}}
  \caption{Diagram for $\Phi^C \otimes \bar\Phi^C (|\psi_n\ket\bra \psi_n|)$
: $U \otimes \bar U$-case}
  \label{page2-1} 
\end{figure}

Fix $\alpha,\beta \in S_{2p}$
and 
let $L(\alpha,\beta)$ be the number of loops (and necklaces).  
Then, the removal $(\alpha,\beta)$ gives the following upper bound for the power of $n$:
\bee\label{L-formula}
\underbrace{-p}_{\text{rescaling}} \;
\underbrace{+L(\alpha,\beta)}_{\text{loops}} \;
\underbrace{-2p-|\alpha^{-1} \beta|}_{\Wg(\alpha^{-1} \beta)} 
= -3p - |\alpha^{-1} \beta| +L(\alpha,\beta) 
\eee

In the case of $U \otimes U$, it is easy to see from Figure \ref{page2-2-2} 
that $L (\alpha,\beta)\leq 3p$
and equality holds if and only if $\alpha = \beta = \id$.
Note that, in this figure, $U$ and $\bar U$ correspond to a solid square and a dotted square, respectively,
so that same-coloured symbols from different kinds of squares will be joined through removals $(\alpha,\beta)$. 
Hence, the necessary condition for (\ref{L-formula}) to be $0$ is 
$\alpha = \beta = \id$.

On the other hand, in case of $U \otimes \bar U$ we have a different bound: $L (\alpha,\beta)\leq 4p$.
This is clear from Figure \ref{page2-1-2} in a similar way. 
Note that equality holds if and only if $\alpha = \id$ and $\beta = \delta$.
In this case, $|\alpha^{-1} \beta| = p$ and $(\ref{L-formula}) =0$.
In order for $(\ref{L-formula})$ to be $0$,
$L (\alpha,\beta)$ may be any number between $3p$ and $4p$ including the edges.
This gap by $p$ gives room for some other $\alpha$'s and $\beta$'s than $\id$ to survive
as $n \rightarrow \infty$.

\subsection{Two different models: $U \otimes  U^*$ and $U \otimes U^T$}
We have so far investigated on
 the pairs $U \otimes \overline U$ and $U \otimes U$.
So, it is natural to ask what would happen for $U \otimes  U^*$ and $U \otimes U^T$.
Here, we again get the flat limiting distribution. 
First, see Figure \ref{page2-left}  where
$n\times \Phi^C \otimes (\Phi^*)^C (|\psi_n\ket  \bra \psi_n| )$ and 
$n\times\Phi^C \otimes (\Phi^T)^C (|\psi_n\ket  \bra \psi_n| )$ are expressed in the graphical calculus. 
Here, $|\phi_n \ket$ is as in (\ref{input-formula}) but 
we set $t=1/k$, i.e. the dimension of input space is exactly $n$. 
This is for convenience as is clear in the Figure \ref{page2-right} below.
\begin{figure}[htbp]
  \centering
  \subfloat[$U \otimes U^*$]{\label{page2-3-1}\includegraphics{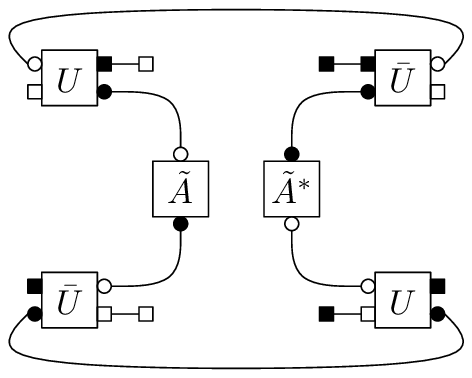}}      
\hfill          
  \subfloat[$U \otimes U^T$]{\label{page2-4-1}\includegraphics{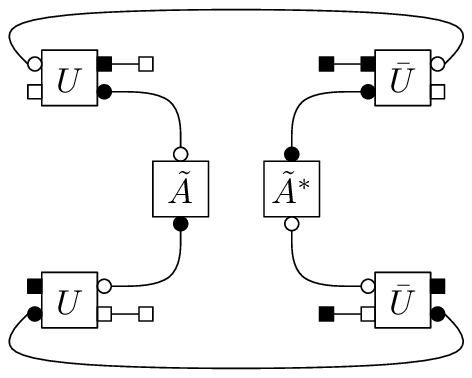}}
  \caption{Diagrams for $U \otimes U^*$ and $U \otimes U^T$-cases }
  \label{page2-left} 
\end{figure}

We again turn Figure  \ref{page2-left} 
into Figure \ref{page2-right} where only the leading wires are drawn. 
\begin{figure}[htbp] 
  \centering
  \subfloat[$U \otimes U^*$]{\label{page2-3-2}\includegraphics{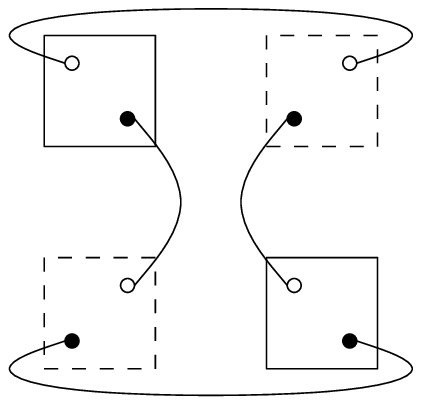}}      
\hfill          
  \subfloat[$U \otimes U^T$]{\label{page2-4-2}\includegraphics{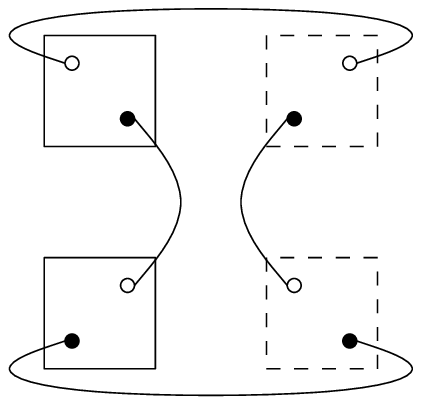}}
  \caption{Leading wires for $U \otimes U^*$ and $U \otimes U^T$-cases }
  \label{page2-right} 
\end{figure}
Again, we have the tight condition $L(\alpha,\beta) \leq 3p$ for those two cases,
which implies that only pair $\alpha =\beta = \id$
survives as $n \rightarrow \infty$. 
Hence, we have the following result.
\begin{theorem}
Under those two assumptions in Theorem \ref{thm:bell-phenomenon},  
the empirical eigenvalue distribution of the matrix $Z_n$ converges 
\emph{almost surely} as $n \rightarrow \infty$ to
$\delta_{1/k^2}$.
Here, $Z_n$ is defined in $U\otimes U^*$ and $U\otimes U^T$ models 
in a similar way as for  $U\otimes \overline U$ and $U\otimes U$ models.
\end{theorem}

\section{Non-pure input states for $\Phi \otimes \bar \Phi$}\label{sec:mixed-inputs}

It is a known fact, following from convexity considerations, that the minimum output entropy of a quantum channel is attained on pure states. However, an interesting by-product of considering mixed inputs is the fact that a product channel $\Phi \otimes \bar \Phi$ and a product of complementary channels $\Phi^C \otimes \bar \Phi^C$ may have outputs with different spectra. This was obviously not the case for pure inputs, since the partial tracing of a pure state on one of the other space does not alter the non-zero spectrum of the resulting density matrix. In this section we consider mixed input states for products of conjugated channels and we study the asymptotic eigenvalue distribution for the output. Our examples do not outperform the Bell state, but some interesting theoretical properties are derived. In the final section, we introduce a new class of \emph{random} input states, correlated to the channel.

\subsection{Mixing a pure Bell state}

We start by considering a fixed-rank perturbation of the Bell state considered in Section \ref{sec:generalized-Bell}. As an input for the product of conjugate channels we consider a state
$$\mathbb M_{n^2}(\mathbb C) \ni X_\text{mixed} = \frac{\mathrm I_l}{l} \otimes |\phi_{n/l}\rangle\langle\phi_{n/l}| \otimes \frac{\mathrm I_l}{l}.$$
The input matrix $X_\text{mixed}$ has rank $l^2$ and we recover the results for the usual Bell states as a particular case by letting $l=1$.

\begin{figure}[htbp]
\centering
\includegraphics{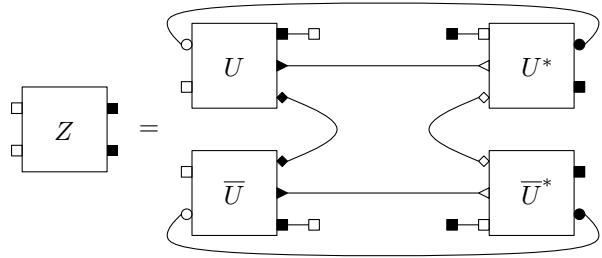}
\caption{Mixed input state for conjugated channels}
\label{fig:mixed-Bell-input}
\end{figure}

We investigate first the product of the direct channels, while complementary channels are discussed in Theorem \ref{thm:mixed-Bell-complementary}. 
This last theorem is the highlight of this section, the output state having in this case \emph{three} different types of eigenvalues, of different multiplicities. 
Note that in this section, for the sake of simplicity, we state our convergence result not in the almost sure setting as we did before, but in the sense of the weak convergence of the empirical 
spectral distribution of output matrices toward some limit deterministic probability measure.

\begin{theorem}\label{thm:mixed-Bell}
In the asymptotic regime $n \to \infty$, $k,l$ fixed, the random matrix  $\mathbb M_{k^2}(\mathbb C) \ni Z = [\Phi^C \otimes \bar \Phi^C](X_\text{mixed})$ has limiting spectral distribution 
$$\frac{1}{k^2}\delta_{\lambda_1} + \frac{k^2-1}{k^2}\delta_{\lambda_2},$$
where
\begin{align*}
\lambda_1 &= \frac{1}{kl^2} + \frac{1}{k^2} - \frac{1}{k^3l^2}\\
\lambda_2 &= \frac{1}{k^2} - \frac{1}{k^3l^2}.
\end{align*}
\end{theorem}
\begin{proof}
Using the graphical Weingarten calculus and counting the loops in the diagrams, we start from
\begin{equation}\label{eq:moments-mixed-Bell}
\mathbb E \mathrm{Tr}(Z^p) = l^{-2p}(n/l)^{-p}\sum_{\alpha, \beta \in \mathcal S_{2p}} n^{\#\alpha} k^{\#(\alpha^{-1}\gamma)} l^{\#\beta} (n/l)^{\#(\beta^{-1}\delta)} \mathrm{Wg}(nk, \alpha^{-1}\beta).
\end{equation}
Using the leading term in the Weingarten function, the general term in the above expression contains $n^{f(\alpha, \beta)}$, where
$$f(\alpha, \beta) = -p + \#\alpha + \#(\beta^{-1}\delta) - 2p - |\alpha^{-1}\beta| = p-(|\alpha| + |\alpha^{-1}\beta| + |\beta^{-1}\delta|).$$
Using the geodesic inequality $|\alpha| + |\alpha^{-1}\beta| + |\beta^{-1}\delta| \geq |\delta| = p$, we conclude that only terms corresponding to permutations on the 
geodesic $\mathrm{id} \to \alpha \to \beta \to \delta$ will contribute asymptotically. We have thus
$$\mathbb E \mathrm{Tr}(Z^p) = (1+o(1)) \sum_{\mathrm{id} \to \alpha \to \beta \to \delta} l^{-2|\beta|} k^{\#(\alpha^{-1} \gamma) - 2p - |\alpha^{-1}\beta|} (-1)^{|\alpha^{-1}\beta|}.$$
Now recall Lemma \ref{countingloops}.
Since in this lemma one has to distinguish the cases where $A = \emptyset$ (i.e. $\alpha = \mathrm{id}$) and $A \neq \emptyset$, we compute separately the following quantities:
\begin{align*}
S_1 &= \sum_{\emptyset = A \subseteq B \subseteq [p]} l^{-2|B|} k^{2 - 2p - |B|} (-1)^{|B|} = k^{2-2p}(1-k^{-1}l^{-2})^p = k^2(k^{-2}-k^{-3}l^{-2})^p.\\
S_2 &= \sum_{\emptyset = A \subseteq B \subseteq [p]} l^{-2|B|} k^{- 2p - |B|} (-1)^{|B|} = k^{-2p}(1-k^{-1}l^{-2})^p = (k^{-2}-k^{-3}l^{-2})^p.\\
S_3 &= \sum_{\emptyset \subseteq A \subseteq B \subseteq [p]} l^{-2|B|} k^{|A| - 2p - |B \setminus A|} (-1)^{|B \setminus A|} = k^{-2p}(1+kl^{-2}-k^{-1}l^{-2})^p \\&= (k^{-2}+k^{-1}l^{-2}-k^{-3}l^{-2})^p.
\end{align*}
(here we used Lemma \ref{multinomialidentities} again).
The result in the statement follows then from 
\begin{align*}
\mathbb E \mathrm{Tr}(Z^p) &= (1+o(1))\left[ S_1 + (S_3 - S_2)\right] \\
&= (1+o(1))\left[(k^{-1}l^{-2} + k^{-2}-k^{-3}l^{-2})^p + (k^2-1)(k^{-2}-k^{-3}l^{-2})^p\right].
\end{align*}
\end{proof}

\begin{theorem}\label{thm:mixed-Bell-complementary}
In the asymptotic regime $n \to \infty$, $k,l$ fixed, the rank $k^2l^2$ random matrix $\mathbb M_{n^2}(\mathbb C) \ni Z = [\Phi \otimes \bar \Phi](X_\text{mixed})$ has limiting spectral distribution
$$\frac{1}{k^2l^2}\delta_{\lambda_1} + \frac{k^2l^2-k^2}{k^2l^2}\delta_{\lambda_2} + \frac{k^2-1}{k^2l^2}\delta_{\lambda_3},$$
where
\begin{align*}
\lambda_1 &= \frac{1}{kl^2} + \frac{1}{k^2l^2} - \frac{1}{k^3l^2}\\
\lambda_2 &= \frac{1}{k^2l^2}\\
\lambda_3 &= \frac{1}{k^2l^2} - \frac{1}{k^3l^2}.
\end{align*}
\end{theorem}
\begin{proof}
As in the previous theorem, we start from the moment formula obtained via graphical calculus
$$\mathbb E \mathrm{Tr}(Z^p) = l^{-2p}(n/l)^{-p}\sum_{\alpha, \beta \in \mathcal S_{2p}} k^{\#\alpha} n^{\#(\alpha^{-1}\gamma)} l^{\#\beta} (n/l)^{\#(\beta^{-1}\delta)} \mathrm{Wg}(nk, \alpha^{-1}\beta).$$
Note that the only difference between the above formula and equation \eqref{eq:moments-mixed-Bell} is that $n$ and $k$ got switched in the first two factors of the general term. 
The exponent of the large parameter $n$ is given by
$$f(\alpha, \beta) = -p + \#(\alpha^{-1}\gamma) + \#(\beta^{-1}\delta) - 2p - |\alpha^{-1}\beta| = p-(|\gamma^{-1}\alpha| + |\alpha^{-1}\beta| + |\beta^{-1}\delta|).$$
The permutations $\alpha,\beta$ which minimize the above quantity are the ones on the geodesic $\gamma \to \alpha \to \beta \to \delta$. This geodesic has been studied in \cite{cn1}, Section 6.1. 
We recall here that these permutations are indexed by pair of subsets $\emptyset \subseteq A \subseteq B \subseteq [p]$ and that one has 
$$\#\alpha = \begin{cases} 2 & A = \emptyset \\ |A| & A \not= \emptyset \end{cases}.$$
A similar formula works for $\beta$ and $B$.
Also, 
\[
|\alpha^{-1}\beta| = |B \setminus A|
\]
One has 
$$\mathbb E \mathrm{Tr}(Z^p) = (1+o(1)) (kl)^{-2p}
\sum_{\emptyset \subseteq A \subseteq B \subseteq [p]} k^{\#\alpha-|B \setminus A|}l^{\#\beta-|B|}(-1)^{|B \setminus A|}.$$
As in the proof of Theorem \ref{thm:mixed-Bell}, we expand the general term in the above sum in terms of $|A|$ and $|B \setminus A|$ and deal separately with the empty set cases.
\begin{align*}
S_1 &= (kl)^{-2p} \sum_{\emptyset \subseteq A \subseteq B \subseteq [p]} k^{|A|} (k^{-1})^{|B \setminus A|} 
= (kl)^{-2p} \l(1+k-\frac{1}{k}\r)^p \\
S_2 &=  (kl)^{-2p} \sum_{\emptyset = A \subseteq B \subseteq [p]}  (k^{-1})^{|B|} 
= (kl)^{-2p} \l(1-\frac{1}{k}\r)^p \\
S_3 &=  (kl)^{-2p} \sum_{\emptyset = A \subseteq B \subseteq [p]}  k^{2-|B|}
= (kl)^{-2p}k^2 \l(1-\frac{1}{k}\r)^p\\
S_4 &= \sum_{\emptyset = A = B } (kl)^{-2p} k^2 =  (kl)^{-2p} k^2\\
S_5 &= \sum_{\emptyset =A=B} (kl)^{-2p} k^2l^2 =(kl)^{-2p} k^2l^2
\end{align*}
Then,
\begin{align*}
\E \trace (Z_n^p) &= (1 + o(1)) 
[S_1 -S_2 + S_3 - S_4 +S_5] \\
&= (1 + o(1)) \frac{1}{(kl)^{2p}} \l[ \l(1+k-\frac{1}{k}\r)^p
(k^2-1) \l(1-\frac{1}{k}\r)^p + (l^2 -1)k^2
\r] \\
&= (1 + o(1)) \Big[ \l(\frac{1}{k^2l^2}+\frac{1}{kl^2}-\frac{1}{k^3l^2}\r)^p \\
&\hspace{3cm} + 
(k^2-1) \l(\frac{1}{k^2l^2}-\frac{1}{k^3l^2}\r)^p + (l^2 -1)k^2 \frac{1}{k^2l^2}
\Big]
\end{align*}
\end{proof}

\subsection{Adapted mixed inputs for $\Phi \otimes \bar \Phi$}\label{sec:adapted}

In the final section of this paper we introduce a new class of input states for product channels which we believe have interesting properties. 
We consider inputs which have the particularity of being \emph{random} and \emph{correlated to the channel}. 
Our approach is motivated by the fact that the quantum state reaching the minimum value in equation \eqref{eq:MOE} depends \emph{a priori} on the channel $\Phi$. 
This obvious fact is hard to exploit in practice, and, to our knowledge, all attempts up to date to find ``interesting'' input states do not use this dependence. 
The states we introduce are \emph{adapted} to the channel, in the sense that they are constructed using the same randomness as the channel (the unitary matrix $U$ appearing in 
the Stinespring form of the quantum channel). In order to do that, we have to make one compromise: the input states we consider are no longer pure, but mixed.

In Figure \ref{fig:mixed-adapted-input} we describe the input matrix $X$. Note that $X$ is defined via the same unitary matrix $U$ appearing in the channels $\Phi$ and 
$\overline \Phi$, hence the adaptedness (or the correlation). The third subfigure contains the global picture for the output matrix $Z$. Note that each unitary matrix $U$ is 
wired to a copy of $U^*$, hence one expects some cancellations to occur and, in the end, to obtain a low-entropy output $Z$. 

One could, in principle, apply the graphical calculus to compute the moments of $Z$. In practice, it turns out that this task is difficult since there are 8 boxes $U$ in 
each group of $Z$. The Weingarten sum associated to a moment $p$ of $Z$ will be indexed by pair of permutations $\alpha,\beta \in \mathcal S_{4p}$, making it very 
difficult to compute. We postpone this computation to further work. 

\begin{figure}[htbp]
  \centering
  \subfloat[Output matrix $Z$ for an input $X$ which may be non-pure.]{\includegraphics{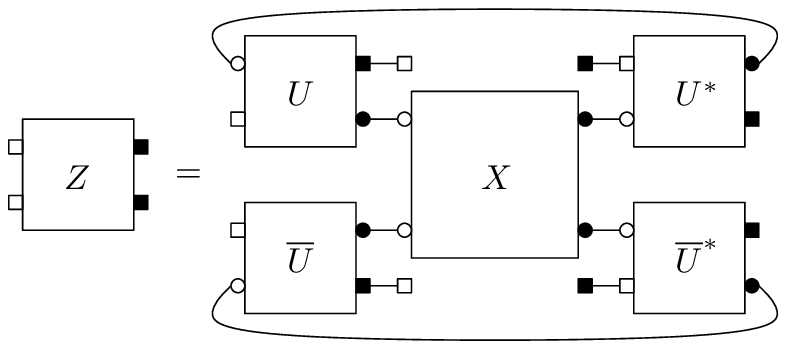}} \\
  \subfloat[An adapted input $X$ of rank $k$.]{\includegraphics{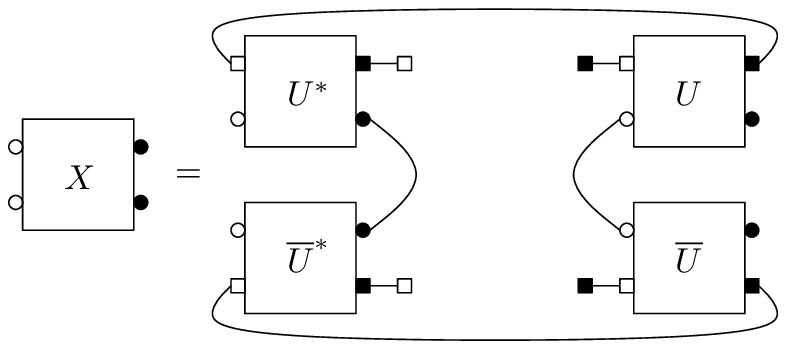}}\\
  \subfloat[The resulting output matrix $Z$ of maximal rank $k^2$.]{\includegraphics{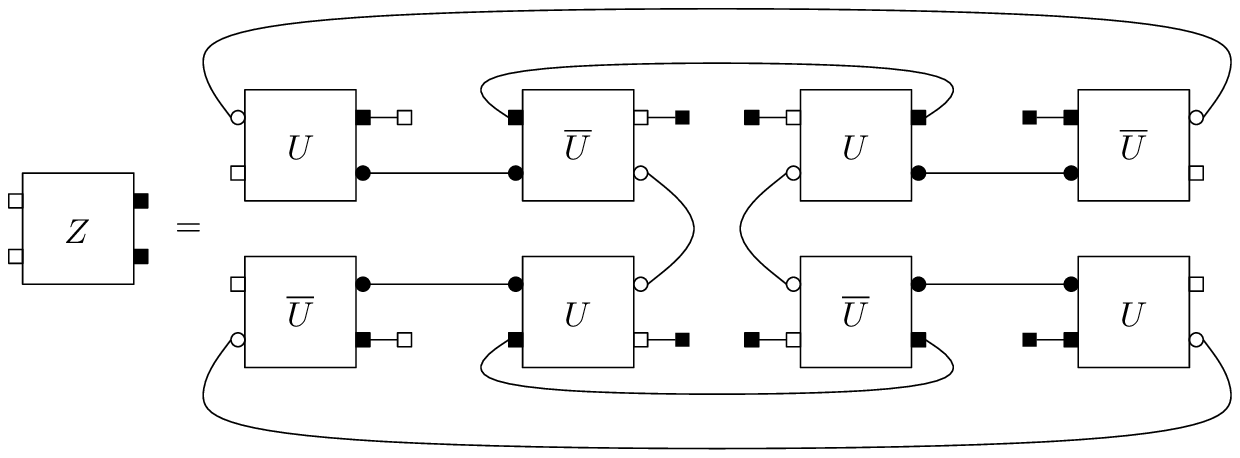}} 
  \caption{Adapted inputs for conjugated channels}
  \label{fig:mixed-adapted-input}
\end{figure}

\section{Discussion}
In this paper, we investigated eigenvalue distribution of
outputs of tensor-product of quantum channels. 

We gathered evidence that
 product of a quantum channel and its complex congregate with a 
 Bell state input
gives the least output entropy.
However this is not proven generally yet. 
In order to get smaller output entropy - if it is possible -
a direction to explore could be to investigate a situation where both the channel and the input are random - and correlated in an appropriate way, as suggested in the adapted input model of Section \ref{sec:adapted}.
We leave this for future work. 

\section*{acknowledgements}
The three authors will like to thank in the first place the ``Quantum Information Theory'' program at the Mittag-Leffler Institute, where this collaboration was initiated. Our research was supported by NSERC Discovery grants
and an ERA at the University of Ottawa (B.C.). The research of I.N. was supported by a PEPS grant from the Institute of Physics of the CNRS and the ANR project ANR 2011 BS01 008 01. 
The research of B.C. was also supported by the ANR Granma and this manuscript was finalized while he was visiting the RIMS at Kyoto University.
The research of M.F. was supported by QuantumWorks and NSERC Discovery grant.


\begin{thebibliography}{99}

\bibitem{bcn}
Belinschi, S., Collins, B. and Nechita, I.
{\it Laws of large numbers for eigenvectors and eigenvalues associated to random subspaces in a tensor product.} 
arXiv:1008.3099.









\bibitem{collins-imrn}
Collins, B. 
{\it Moments and Cumulants of Polynomial random variables on unitary groups, 
the Itzykson-Zuber integral and free probability }
Int. Math. Res. Not., (17):953-982. 

\bibitem{cn1}
Collins, B. and Nechita, I.
{\it Random quantum channels I: Graphical calculus and the Bell state phenomenon.} 
Comm. Math. Phys. 297 (2010), no. 2, 345-370.

\bibitem{cn2}
Collins, B. and Nechita, I.
{\it Random quantum channels II: Entanglement of random subspaces, R\'enyi entropy estimates and additivity problems.} 
Advances in Mathematics 226 (2011), 1181–1201.

\bibitem{cn3}
Collins, B. and Nechita, I.
{\it Gaussianization and eigenvalue statistics for Random quantum channels (III)}
to appear in Annals of Applied Probability.

\bibitem{cn-entropy}
Collins, B. and Nechita, I.
{\it Eigenvalue and Entropy Statistics for Products of Conjugate Random Quantum Channels.}
Entropy, 12(6), 1612-1631.


\bibitem{collins-nechita-zyczkowski}
Collins, B., Nechita, I.; \.Zyczkowski, K.
{\it Random graph states, maximal flow and Fuss-Catalan distributions.}
J. Phys. A: Math. Theor. 43, 275303.

\bibitem{collins-sniady}
Collins, B. and \'Sniady, P.
{\it Integration with respect to the Haar measure on unitary, orthogonal and symplectic group.} 
Comm. Math. Phys. 264, no. 3, 773--795. 

\bibitem{Fuk06} 
M. Fukuda,
``Simplification of additivity conjecture in quantum information theory'',
{\em Quant. Info. Proc.}, {\bf 6}, 179--186, (2007);
arXiv:quant-ph/0608010.

\bibitem{fukuda-king}
Fukuda, M. and King, C.
{\it Entanglement of random subspaces via the Hastings bound}
arXiv:0907.5446

\bibitem{fukuda-king-moser}
Fukuda, M., King, C. and Moser, D.
{\it Comments on Hastings' Additivity Counterexamples.}
arXiv:0905.3697.

\bibitem{FukudaWolf07}
M. Fukuda, M. M. Wolf, 
``Simplifying additivity problems using direct sum constructions'', {\em J. Math. Phys. } {\bf 48} 072101, (2007).




\bibitem{hastings}
Hastings, M.B.
{\it Superadditivity of communication capacity using entangled inputs}
Nature Physics 5, 255. 

\bibitem{hayden-winter}

Hayden, P. and Winter, A.
{\it Counterexamples to the maximal p-norm multiplicativity conjecture for all $p>1$}. 
Comm. Math. Phys. 284, no. 1, 263--280.

\bibitem{Holevo05}
Holevo, A. S. 
``On complementary channels and the additivity problem'',
Probab. Theory and Appl., {\bf 51}, 133-143, (2005).

 
\bibitem{KMNR}
C. King, K. Matsumoto, M. Nathanson, M. B. Ruskai,
``Properties of Conjugate Channels with Applications to Additivity and Multiplicativity'', 
Markov Processes and Related Fields, volume {\bf 13}, no. 2, 391 -- 423 (2007).



\bibitem{nica-speicher}
Nica, A and Speicher, R.
{\it Lectures on the combinatorics of free probability} 
volume 335 of London Mathematical Society Lecture Note Series. Cambridge University Press, Cambridge.


\bibitem{Shor04}
P. W. Shor,
``Equivalence of additivity questions in quantum information theory'',
Comm. Math. Phys. {\bf 246}(3):453-472 (2004).





\end{thebibliography}
\end{document}